\documentclass[11pt, a4paper, twocolumn]{IEEEtran}
\usepackage{amsthm}
\usepackage{amsmath}
\usepackage{amssymb}
\usepackage{hyperref}
\usepackage{enumitem}
\usepackage{color}
\usepackage{xcolor}
\usepackage{mathrsfs}
\usepackage{mathtools}
\usepackage{xfrac}
\usepackage{marginnote}

\newtheorem{thm}{Theorem}

\newtheorem{cor}[thm]{Corollary}
\newtheorem{prop}[thm]{Proposition}
\newtheorem{lem}{Lemma}
\newtheorem{model}{Model}

\theoremstyle{definition}
\newtheorem{defi}{Definition}

\newcommand{\note}[1]{\textcolor{blue}{#1}}

\newcommand*\dif{\mathop{}\!\mathrm{d}}

\newcommand{\A}{\mathcal{A}}

\newcommand{\R}{\mathbb{R}}
\newcommand{\C}{\mathbb{C}}
\newcommand{\N}{\mathbb{N}}
\renewcommand{\i}{\mathrm{i}}

\newcommand{\ind}{\in [N]}
\newcommand{\inn}{\in [n]}
\newcommand{\inzn}{\in [2n]}

\newcommand{\norm}[1]{\left\Vert #1\right\Vert}			
\newcommand{\dualnormsymbol}{^{\circ}}			
\newcommand{\smallnorm}[1]{\lVert #1\rVert}				

\newcommand{\normpsi}[2]{\norm{#1}_{\psi_{#2}}}	

\newcommand{\abs}[1]{\left| #1\right|}					
\newcommand{\smallabs}[1]{\lvert #1\rvert}				

\newcommand{\br}[1]{\left(#1\right)}					
\newcommand{\bigbr}[1]{\big(#1\big)}					
\newcommand{\Bigbr}[1]{\Big(#1\Big)}					

\newcommand{\set}[1]{\left\lbrace #1\right\rbrace}		
\newcommand{\smallset}[1]{\lbrace #1\rbrace}			
\newcommand{\bigset}[1]{\bigl\{#1\bigr\}}				

\newcommand{\scp}[2]{\left\langle #1,#2\right\rangle} 	
\newcommand{\smallscp}[2]{\langle #1,#2\rangle} 		

\newcommand{\prob}[1]{\mathbb{P}[\; #1\;]} 	
\newcommand{\expec}[1]{\mathbb{E}[ #1]} 	

\newcommand{\fro}{\text{F}}

\DeclareMathOperator{\diag}{diag}
\DeclareMathOperator{\Id}{Id}
\DeclareMathOperator{\re}{Re}
\DeclareMathOperator{\im}{Im}
\DeclareMathOperator{\trace}{tr}
\DeclareMathOperator{\vc}{vec}

\newcommand{\herm}{*}
\renewcommand{\vc}[1]{\ensuremath{\boldsymbol{#1}}}

\author{Fabian Jaensch and Peter Jung}

\title{Robust Recovery of Sparse Nonnegative Weights from Mixtures of Positive-Semidefinite Matrices}
\begin{document}
\maketitle
\begin{abstract}
  We consider a structured estimation problem where an observed matrix
  is assumed to be generated as an $s$-sparse linear combination of
  $N$ given $n\times n$ positive-semidefinite matrices.  Recovering
  the unknown $N$-dimensional and $s$-sparse weights from noisy
  observations is an important problem in various fields of signal
  processing and also a relevant pre-processing step in covariance
  estimation. We will present related recovery guarantees and focus on
  the case of nonnegative weights. The problem is formulated as a
  convex program and can be solved without further tuning. Such
  robust, non-Bayesian and parameter-free approaches are important for
  applications where prior distributions and further model parameters
  are unknown.  Motivated by explicit applications in wireless
  communication, we will consider the particular rank-one case, where
  the known matrices are outer products of iid. zero-mean subgaussian
  $n$-dimensional complex vectors. We show that, for given $n$ and
  $N$, one can recover nonnegative $s$--sparse weights with a
  parameter-free convex program once $s\leq O(n^2 / \log^2(N/n^2)$.
  Our error estimate scales linearly in the instantaneous noise power
  whereby the convex algorithm does not need prior bounds on the
  noise.  Such estimates are important if the magnitude of the
  additive distortion depends on the unknown itself.
\end{abstract}
\section{Introduction}

In compressed sensing one is confronted with the inverse problem of
recovering unknown but structured signals from only few observations,
far less the dimension of the signal. This methodology is reasonable
if the effective dimension of the signals of interest is much smaller
then its ambient dimension. A prominent example is the set of
$s$--sparse and $N$-dimensional vectors where $s\ll N$. The original
recovery problem has combinatorial nature and is computationally infeasible
since one essentially has to implicitly search over the exponentially
$\binom{N}{s}$ many support combinations of the unknown signal.

The first fundamental theoretical breakthroughs \cite{Candes2005b,candes:stablesignalrecovery,Donoho2006a} show that
for a linear and real-valued measurement model and under further
assumptions on the so called measurement matrix, it is possible to
recover the unknown vector in the noiseless case by a linear program.
In the noisy case it is also possible to obtain provable guarantees
for certain convex programs, see here for example \cite{Foucart2013},
which usually require a tuning parameter
that often depends on further properties on the noise contribution,
in most cases, the $\ell^2$--norm of the noise. However, there are several
signal processing problems where it is difficult to aquire this
knowledge. For example, in the Poisson noise model this depends on the
unknown signal itself. Another example are certain applications in
sparse covariance matching where the error contribution comes from the
deviation of the empirical to the true covariance matrix, which in turn
depends on the sparse parameter to recover. There are some
concepts known in the literature how to deal with convex compressed
sensing programs in the absence of this a-priori information. To
mention some examples, the quotient bounds \cite{Wojtaszczyk2010} of
the measurement matrix can provide guarantees for the basis pursuit or
the basis pursuit denoising, see for example also \cite[Chapter
11]{Foucart2013}. Empirical approaches and modifications of the convex
programs are also known to get rough estimates for the noise power,
see for example \cite{herrmann:wsa18}. Interestingly, it has been observed
also in \cite{Donoho92,Bruckstein,Slawski,Meinshausen2013} that nonnegativity of the
unknowns together with particular properties of the measurement matrix
yield a ``self-tuning'' approach, which has been worked out in
\cite{Kabanava2015} for the nuclear norm and in \cite{Kueng2018} for the
$\ell^1$--norm with respect to guarantees formulated in the
terminology of the robust nullspace property.

\section{Main Results}
Motivated by covariance matching problems, briefly also sketched
below, we shall consider the problem of recovering
nonnegative and sparse vectors from the noisy matrix
observation
\begin{equation}
  Y=\A(x)+E,
  \label{eq:measmodel}
\end{equation}
where $\A : \R^N \mapsto \C^{n\times n}$ is a given linear measurement
map.
We establish recovery guarentees for the generic convex
program
\begin{equation}
  x^\sharp=\arg\min_{z\geq 0}\|\A(z)-Y\|,
  \label{eq:generic:nnnorm}
\end{equation}
where $\norm{\cdot}$ is a given norm on $\C^{n\times n}$.  We shall
write $\norm{\cdot}_p$ for the $\ell^p$-norms for vectors and matrices
(when seen as a vector). In particular, for the Frobenius norm
$\|\cdot\|_{\fro}=\|\cdot\|_2$ the problem \eqref{eq:generic:nnnorm}
is the so called {\em Nonnegative Least-Squares} (NNLS). Guarantees
for this case have been established already in
\cite{Slawski,Meinshausen2013,Kueng2018}. See here also
\cite{Kabanava2015} for a similar approach for the low-rank and
positive-semidefinite case (instead of sparse and elementwise
nonnegative).

In this work we follow techniques established mainly in
\cite{Kueng2018} and extend our work to the special matrix-valued
observation model \eqref{eq:measmodel}. Then we investigate a
structured random measurement map $\A$ which can be represented as a
matrix with independent heavy-tailed columns containing vectorized
outer products of random vectors with itself. Such matrices are
sometimes also called as (self-) Khatri-Rao products. By construction
such random matrices are biased which is essential for the recovery of
nonnegative vectors via \eqref{eq:generic:nnnorm} (further details
below or see for example also the discussion \cite{Shadmi:isit19} for
the unstructured case).  Recent results about the RIP property of such
matrices after centering and in the real case have obtained in
\cite{Fengler:krrip:2019}.  These investigations have been worked out
towards a NNLS recovery guarantee in
\cite{Hag:isit18,alex2019nonbayesian} using the nullspace property in
the special case where the vectors are drawn from the complex
sphere. In this work we focus on the complex subgaussian case and
establish the corresponding compressed sensing recovery guarantee.

To state our main results we need the following definitions.
For the case of a generic norm $\norm{\cdot}$ on $\C^{n\times n}$ we
let $\norm{\cdot}\dualnormsymbol$ be the corresponding dual norm
defined as
\begin{equation*}
  \norm{Y}\dualnormsymbol:=\sup_{\norm{X}\leq 1}\langle Y,X\rangle,
\end{equation*}
where by $\scp{X}{Y}:=\trace{X^*Y}$ we denote the Hilbert-Schmidt (Frobenius)
inner  product on $\C^{n\times n}$. To simplify notation we will stick
to square matrices in the space $\C^{n\times n}$ but the first part of this work
can be easily rewritten for the non-square case or even for a generic
inner product space.
\cite[Definition 4.21]{Foucart2013} is essential for our analysis.
\begin{defi}\label{def:robust-NSP}
  Let $q\in [1,\infty)$ and $s\in\N$. We say that a linear map $\A : \R^N \mapsto \C^{n\times n}$
  satisfies the $\ell^q$-\textit{robust nullspace property ($\ell^q$-NSP)} of
  order $s$ with respect to $\norm{\cdot}$ with parameters 
  $\rho\in(0,1)$ and $\tau>0$ if for all
  $S\subseteq [N]\coloneqq\set{1,\ldots,N}$ with cardinality $|S|\leq s$
  \begin{equation}\label{eq:NSP}
    \norm{v_S}_q	\leq	\frac{\rho}{s^{1-1/q}} \norm{v_{S^c}}_1 + \tau \norm{\A(v)}
  \end{equation}
  holds for all $v \in \R^N$. Here, $v_S\in\R^N$ denotes the vector
  containing the same entries as $v$ at the positions in $S$ and zeros
  at the others and $S^c=[N]\backslash S$.
\end{defi}
Furthermore, by $\sigma_s(x)_1=\min_{|S|\leq s}\|x-x_S\|_1$ we denote
here the best $s$-term approximation to $x\in\R^N$ in the
$\ell^1$-norm.  The nullspace property is essential for recovery via
$\ell^1$-based convex recovery programs like basis pursuit and basis
pursuit denoising, see for example \cite[Theorem
4.22]{Foucart2013}. When recovering nonnegative vectors, the following
additional property, often called ${\cal M}^+$-criterion, controls
the $\ell^1$-norms of all feasible vectors such that
$\ell^1$-regularization becomes superfluous.
\begin{defi}
  A linear map $\A : \R^N \mapsto \C^{n\times n}$ satisfies the
  ${\cal M}^+$-criterion if there exists a matrix $T\in\C^{n\times n}$
  such that $w\coloneqq\A^*(T)>\vc{0}$ componentwise. For a given $T$,
  we then define the condition number
  $\kappa(w)=\max_{i\in[N]}|w_i|/\min_{i\in[N]}|w_i|$.  
  \label{def:kappa}
\end{defi}
{Note that $\kappa(w)$ is scale-invariant, i.e.,
  $\kappa(w)=\kappa(t w)$ for all $t\neq 0$.}  For further
illustration of this property, consider the noiseless setting and
assume for simplicity that we can find $T$ such that
$w=\A^*(T)=(1,\dots,1)=:1_N$ is the all-one vector. Then
\begin{equation*}
  \begin{split}
    \|x\|_1
    &\overset{x\geq 0}{=}\langle 1_N,x\rangle=\langle\A^*(T),x\rangle\\
    &\,\,=\langle T,\A(x)\rangle=\langle T,Y\rangle=\text{const}
  \end{split}
\end{equation*}
shows that all feasible vectors $x$ have the same $\ell^1$-norm.  As
we shall show below, a similar conclusion follows for the general case
$w>0$ and the tightness of such an argument will depend on $\kappa(w)$.

The following theorem essentially extends and refines \cite[Theorem
3]{Kueng2018} to the case of matrix observations and generic norms.
\begin{thm}
  Let $q\geq1$ and let $\A: \R^N \rightarrow \C^{n\times n}$ be a linear map which  (i) satisfies the
  $\ell^q$--NSP of order $s$ with respect to $\|\cdot\|$ and with
  parameters $\rho\in[0,1)$ and $\tau>0$ and (ii) fulfills the
  ${\cal M}^+$-criterion for $T\in\C^{n\times n}$ with
  $\kappa=\kappa(\A^*(T))$.  If $\rho\kappa<1$, then, for any nonnegative $x\in\R^{N}$
  and all $E\in\C^{n\times n} $, the solution $x^\sharp$ of \eqref{eq:generic:nnnorm} for
  $Y=\A(x)+E$ obeys
  \begin{equation}
    \|x^\sharp-x\|_p	\leq	\frac{C'\kappa\sigma_s(x)_1}{s^{1-\frac{1}{p}}}
    +  \frac{D'\kappa}{s^{\frac{1}{q}-\frac{1}{p}}}\br{\tau + \frac{\theta}{s^{1-\frac{1}{q}}}}\norm{E}
    \label{eq:thm:main:nonneg}
  \end{equation}
  for all $p\in[1,q]$, where
  $C' \coloneqq 2\frac{(1+\kappa\rho)^2}{1-\kappa\rho}$ and
  $D' \coloneqq 2\frac{3 + \kappa\rho}{1 -
    \kappa\rho}$ and $\theta={\smallnorm{\A^*({T})}_\infty^{-1}}\cdot{\norm{T}\dualnormsymbol}$.
  \label{thm:main:nonneg}
\end{thm}
We prove this theorem in Section \ref{sec:nonneg:generic}.
As a second main result, we show that it is applicable to the
following random observation model:
\begin{model}
  Let $a_i= (a_{i,k})_{k\in[n]}\in\C^n$ for $i=1,\ldots, N$ be independent
  random vectors with independent subgaussian real and
  imaginary parts $\re(a_{i,k})$ and $\im(a_{i,k})$ satisfying
  \begin{align*}
    \expec{a_{i,k}}
    &=\expec{\re(a_{i,k})} =\expec{\im(a_{i,k})} = 0 \\
    \expec{\re(a_{i,k})^2} &= \expec{\im(a_{i,k})^2} = 1/2,
  \end{align*}
  so that $\expec{|a_{i,k}|^2}=1$ and $\expec{\norm{a_i}_2^2}=n$.
  We consider the following map $\A:\R^N\rightarrow\C^{n\times n}$:
  \begin{equation}
    \A(x):=\sum_{i=1}^N x_i a_i a_i^\herm
    \label{eq:rankone:model}
  \end{equation}
  Let $\psi_2\geq 1$ be a uniform bound on the subgaussian norms
  $\normpsi{\re(a_{i,k})}{2}$ and $\normpsi{\im(a_{i,k})}{2}$ for all $i\in [N],k\in
  [n]$, see \eqref{eq:def:psip} below for the
  definition.
  \label{model:aa}
\end{model}
The case where the vectors $a_i$ are drawn uniformly from the complex sphere has been
discussed already in \cite{Hag:isit18} and the full proof of the
recovery guarantee can be found in \cite{alex2019nonbayesian}.
In this work we discuss the subgaussian iid case instead where
additionally also the distribution of $\|a_i\|_2$ affects the probability bounds.
We have the following second main result.
\begin{thm}
  Let $\A: \R^N \rightarrow \C^{n\times n}$ be a random measurement
  map following  Model \ref{model:aa}. Set
  $m\coloneqq2n(n-1)$ and assume
  \begin{equation}
    s\lesssim m\log^{-2}(N/m),
    \label{eq:thm:main:subgaussian:phasetransition}
  \end{equation}
  $n\gtrsim \log(N)$ and $N\geq m$.
  With probability at least $1-4\exp(-c_1\cdot n)$ it holds that for
  all $p\in[1,2]$, all $x\in\R_{\geq 0}^N$ and $E\in\C^{n\times n}$, the
  solution $x^\sharp$ of the NNLS (the convex program \eqref{eq:generic:nnnorm}
  for the Frobenius norm $\|\cdot\|_\fro$) for $Y=\A(x)+E$ obeys
  \begin{equation}\begin{split}
      \|&x^\sharp-x\|_p
      \leq	\frac{c_2\sigma_s(x)_1}{s^{1-\frac{1}{p}}}
      +  \frac{c_3\br{c_4 +
          \sqrt{\frac{n}{s}}}}{s^{\frac{1}{2}-\frac{1}{p}}}\frac{\norm{E}_\fro}{n},
    \end{split}
    \label{eq:thm:main:subgaussian}
  \end{equation}
  where $C_1,c_1,c_2,c_3,c_4$ are absolute constants depending on $\psi_2$
  but not on the dimensions. 
  \label{thm:main:subgaussian}
\end{thm}
The proof of this theorem will be presented in Section
\ref{sec:nonneg:subgaussian}. We have not optimized the constants but
some concrete numbers are for example $c_2=11.36$, $c_3=15.55$ and
$c_4=3.07$, more details are in the proof below. The constants $C_1$ and $c_1$
depend on the subgaussian norm $\psi_2$ in Model \ref{model:aa} and
can also be obtained from the proof.

\subsection{Motivating Application}
We will briefly mention an application of the results above in the area of
wireless communication \cite{Hag:isit18,Che2019,alex2019nonbayesian}. An important task in
wireless networks is to estimate the nonnegative large-scale
path-gain coefficients (product of transmit power and attenuation due
to path-loss) and user activity using multiple antennas.  Here, a small
subset of $s\ll N$ devices indicate activity by transmitting specific
length-$n$ sequences which superimpose at each receive antenna with
individual and unknown instantaneous channel coefficients. Let us
denote this nonnegative vector of large-scale path-gains by
$\gamma\in\R^N$ and due to activity $\gamma$ is essentially
$s$--sparse. For a single receive antenna, the received (noiseless)
signal would be:
\begin{equation*}
  y=A\diag(\sqrt{\gamma})h
\end{equation*}
Here $h\in\C^N$ is the vector of unknown small-scale fading
coefficients and $A=(a_1|\dots|a_N)\in\C^{n\times N}$ is the matrix
containing all the sequences $a_i$ registered in the network (in real
applications for example pseudo-random sequences seeded by the device
id). Well-known results in compressed sensing show  that when using
sufficiently random sequences  of length
$n \simeq s\cdot \text{polylog}(N)$ for given $s$ and $N$, one can recover per antenna
w.h.p. the complex-valued channel coefficients $\diag(\sqrt{\gamma})h$
and the activity pattern (the essential support).

However, since in future even the number of active devices $s$ will
grow considerably, the question is how to further gain from a massive
number of receive antennas {\em when one is only interested in
  reconstructing $\gamma$ or its support}.
A very promising approach is to recover then the sparse and non-negative vector
from covariance information, an approach which has been investigated
already in \cite{Pal2015d}.

In more detail, assuming that the
small-scale fading coefficient vectors $h$ for different receive
antennas and different users are uncorrelated, we can view the
received signal $y$ at each receive antenna as a new realization of
the same random process having a covariance matrix which is
parametrized by $\gamma$, i.e., this leads precisely to the following
covariance model:
\begin{equation*}
  \A(\gamma)=\mathbb{E} yy^*=A\diag(\gamma)A ^*
\end{equation*}
Here $\gamma$ is an unknown nonnegative and sparse parameter which
should match (in a reasonable norm) the empirical covariance 
\begin{equation*}
  Y=\frac{1}{M}\sum_{k=1}^My_ky_k^*\overset{(!!)}{=}\A(\gamma)+E
\end{equation*}
computed from the received vectors $\{y_k\}_{k=1}^M\subset\C^n$ at $M$
receive antennas. The error $E$ accounts therefore for the fact of
having only finite $M$ (and obviously further unknown disturbances
like adversarial noise and interference always present in communication
systems). Note that the error $E$ above usually depends then on the
unknown parameter $\gamma$ as well.

Our result, Theorem \ref{thm:main:subgaussian}, now shows that
pathloss coefficients and activity of up to
$s\leq O(n^2 / \log^2(N/n^2))$ devices can be robustly recovered from
the empirical covariance $Y$ over sufficiently many receive antennas
when matching the model in the Frobenius norm. Note that, although not
further considered in this work, errors due to having finite $M$ will
vanish with increasing $M$ for moderate assumptions on the
distribution of $h$ and one could obviously also make concrete
statements about the concentration of $\|E\|_\fro$ in
\eqref{eq:thm:main:subgaussian} in terms of $M$, see
\cite{alex2019nonbayesian}.

\section{Generic Nonnegative Recovery Guarantee via the Nullspace Property}
\label{sec:nonneg:generic}

In this section, we are following \cite{Kueng2018} aiming towards
showing Theorem \ref{thm:main:nonneg} which is a more general and
refined version of the deterministic guarantee given in
\cite{Kueng2018}.  The proof of Theorem \ref{thm:main:nonneg} is given
at the end of this section. First, we will need \cite[Theorem
4.25]{Foucart2013}:
\begin{thm}[Theorem 4.25 in \cite{Foucart2013}]\label{thm:1}
  Let $q \in [1,\infty)$ 
  and $s \in \N$. Assume $\A$ satisfies the $\ell^q$-NSP of order $s$
  with respect to $\norm{\cdot}$ and with constants
  $\rho \in (0,1)$ and  $\tau > 0$. Then, for any $p \in [1,q]$ and for all
  $x,z \in \R^N$,
  \begin{align}
    \norm{x - z}_p
    &\leq \frac{C(\rho)}{s^{1-1/p}} \br{\norm{z}_1 - \norm{x}_1 + 2 \sigma_s(x)_1} \nonumber\\
    &+ D(\rho)\tau s^{1/p - 1/q} \norm{\A (x-z)}
      \label{eq:thm:1}
  \end{align}
  holds, where $C(\rho) \coloneqq \frac{( 1 + \rho )^2}{1 - \rho}\leq
  \frac{(3 + \rho)}{1 - \rho}=:D(\rho)$.
\end{thm}
First we show a modified version of \cite[Lemma 5]{Kueng2018}.  Recall
that for a diagonal matrix
$W = \diag(w_1, \ldots, w_N) \in \R^{N\times N}$ considered as a
linear operator from $\R^N$ to $\R^N$ equipped with $\norm{\cdot}_p$
for any $p \in [1, \infty ]$, the operator norm is given as
$\norm{W}_o = \max \set{\abs{w_1}, \ldots, \abs{w_N}}$. Furthermore,
$W$ is invertible if and only if all the diagonal entries are nonzero,
with inverse
$W^{-1} = \diag (\frac{1}{|w_1|}, \ldots,
\frac{1}{|w_N|})$.
Thus, in this case we can also write the condition number
in Definition \ref{def:kappa} as $\kappa(w)=\|W\|_o\|W^{-1}\|_o$.

\begin{lem}\label{lem:1}
  Let $W = \diag(w)$ for some
  $0<w\in\R^N$. If $\A$ satisfies the assumption in
  Theorem \ref{thm:1} and
  $\kappa=\kappa(w) <
  \frac{1}{\rho}$, then $\A \circ W^{-1}$ satisfies the
  $\ell^q$-NSP of order $s$ with respect to $\norm{\cdot}$ and
  with constants $\tilde{\rho} \coloneqq \kappa \rho$,
  $\tilde{\tau} \coloneqq \norm{W}_o \tau$.
\end{lem}
\begin{proof}
  Let $S \subseteq [N]$ with $|S|\leq s$ and $v \in \R^N$. Since $W$ is diagonal, we have $(W^{-1}v)_S=W^{-1}v_S$ (same for $S^c$). We get:
  \begin{align*}
    &\norm{v_S}_q	\leq\norm{W}_o \norm{(W^{-1}v)_S}_q	\\
    &\leq	\norm{W}_o\br{\frac{\rho}{s^{1-1/q}}\norm{(W^{-1}v)_{S^c}}_1 + \tau \norm{\A (W^{-1}v)}	} \\
    &\leq \frac{\norm{W}_o\norm{W^{-1}}_o\rho}{s^{1-1/q}}\norm{v_{S^c}}_1	+  \norm{W}_o \tau \norm{\A(W^{-1}v)} \\
    &= \frac{\tilde{\rho}}{s^{1-1/q}}\norm{v_{S^c}}_1	+	\tilde{\tau}\norm{(\A\circ W^{-1})v}
  \end{align*}
\end{proof}
The next lemma is a generalization of \cite[Lemma 6]{Kueng2018}.
\begin{lem}\label{lem:2}
  Assume $w \coloneqq \A^*(T) \in \R^N$ is strictly positive for some
  $T \in \C^{n \times n}$ and set $W \coloneqq \diag(w)$. For any nonnegative $x,z \in \R^N$ it holds that
  \begin{equation*}
    \norm{Wz}_1 - \norm{Wx}_1	\leq \norm{T}\dualnormsymbol\norm{\A (z-x)}.
  \end{equation*}
\end{lem}
\begin{proof}
  Let $x,z \in \R^N$ be nonnegative. By construction, we have $W = W^*$ and $Wx$ is nonnegative. This implies
  \begin{align*}
    \norm{Wz}_1	&=	\scp{1_N}{Wz}	=	\scp{w}{z}	\\
                &=	\scp{\A^*(T)}{z}	=	\scp{T}{Az}
  \end{align*}
  where $1_N$ denotes the vector in $\R^N$ containing only ones. With an analogous reformulation for $x$ we get
  \begin{align*}
    \norm{Wz}_1 - \norm{Wx}_1	&=	\scp{T}{\A (z-x)}\\
                                &\leq	\norm{T}\dualnormsymbol\norm{\A(z-x)}.
  \end{align*}
\end{proof}
We can now show a more general version of
\cite[Theorem 3]{Kueng2018} which holds for general $p\in [1,\infty)$
and generic norms on matrices. It parallels Theorem \ref{thm:1} in
the nonnegative case.
\begin{thm}\label{thm:M+}
  Suppose that $\A$ satisfies the assumptions in Theorem
  \ref{thm:1} and that there exists some ${T} \in \C^{n\times n}$ such that
  $\A^*({T})$ is strictly positive. Set $W\coloneqq\diag(\A^*(T))$ and $\kappa\coloneqq\kappa(\A^*(T))=\norm{W}_o\norm{W^{-1}}_o$. If  $\kappa\rho < 1$, then
  \begin{align*}
    \norm{z-x}_p
    &\leq
      \frac{2C(\kappa\rho)\kappa}{s^{1-1/p}}\sigma_s(x)_1	\\
    &+	\frac{D(\kappa\rho)}{s^{1/q-1/p}}
    \Bigbr{\kappa\tau +
      \frac{\norm{W^{-1}}_o \norm{T}\dualnormsymbol }{s^{1-1/q}} }\norm{\A (z-x)}
  \end{align*}
  holds for all $p \in [1,q]$ and all nonnegative $x,z \in \R^N$. 
\end{thm}

Note that we used here the definition of $C(\rho)$ and $D(\rho)$ from
Theorem \ref{thm:1}. Using this result for $p=q=2$ and with
$s^{1-1/q}\geq 1$ yields essentially \cite[Theorem 3]{Kueng2018}.
\begin{proof}
  Let $p \in [1,q]$ and $x,z \in \R^N$ be nonnegative. By Lemma \ref{lem:1}, $\A\circ W^{-1}$ satisfies the NSP with parameters $\tilde{\rho}=\kappa\rho$ and
  $\tilde{\tau}=\|W\|_o\tau =\norm{W}_o\tau$.  Therefore,
  we can now use Theorem \ref{thm:1} for $Wx$ and $Wz$ (instead of $x$
  and $z$) and $\A\circ W^{-1}$ (instead of $\A$):
  \begin{align*}
    \|W&(z-x)\|_p \\
    \overset{\eqref{eq:thm:1}}{\leq}&	C(\kappa\rho)\frac{\norm{Wz}_1 - \norm{Wx}_1 +2\sigma_s(Wx)_1}{s^{1-1/p}}\\
       &+D(\kappa\rho)\norm{W}_o\tau s^{1/p-1/q}\norm{\A(z-x)}
  \end{align*}
  By Lemma \ref{lem:2} and invoking $\sigma_s(Wx)_1 \leq \norm{W}_o \sigma_s(x)_1$, this is at most
  \begin{align*}
    &C(\kappa\rho)\frac{\|{T}\|\dualnormsymbol\norm{\A (z-x)} +2\norm{W}_o\sigma_s(x)_1}{s^{1-1/p}}\\
    +&D(\kappa\rho)\norm{W}_o\tau s^{1/p-1/q}\norm{\A(z-x)}
  \end{align*}
  which we can further upper bounded by using $C(\kappa\rho)\leq D(\kappa\rho)$ by:
  \begin{align*}
    &2 C(\kappa\rho)\norm{W}_o\frac{\sigma_s(x)_1}{s^{1-1/p}}\\
    +&D(\kappa\rho)s^{1/p-1/q}\br{\norm{W}_o\tau +
      \frac{\norm{T}\dualnormsymbol}{s^{1-1/q}}}\norm{\A(z-x)}.
  \end{align*}
  This yields
  \begin{align*}
    &\norm{z-x}_p\leq\|W^{-1}\|_o\norm{W(z-x)}_p\\
    &\leq	2C(\kappa\rho)\kappa\frac{\sigma_s(x)_1}{s^{1-1/p}}\\
    &+ \frac{D(\kappa\rho)}{s^{1/q-1/p}}\Bigbr{\kappa\tau +
      \frac{\norm{W^{-1}}_o\norm{T}\dualnormsymbol}{s^{1-1/q}}}\norm{\A(z-x)}
  \end{align*}
\end{proof}
%

\begin{proof}[Proof of Main Theorem  \ref{thm:main:nonneg}]
  Applying Theorem \ref{thm:M+} above for $Y=\A(x) + E$
  we obtain
  \begin{align*}
    &\|x^\sharp-x\|_p	
        \leq	\frac{2C(\kappa\rho)\kappa}{s^{1-1/p}}\sigma_s(x)_1\\
     &+
        \frac{D(\kappa\rho)}{s^{1/q-1/p}}\Bigbr{\kappa \tau + \frac{ \norm{W^{-1}}_o\norm{T}_o}{s^{1-1/q}}}
        \left(\norm{\A (z)-Y}+\norm{E}\right).
  \end{align*}
  This indeed suggests to use the convex program \eqref{eq:generic:nnnorm} for recovery.
  Obviously, the minimizer $x^\sharp$ of \eqref{eq:generic:nnnorm} fulfills
  \begin{equation*}
    \|A(x^\sharp)-Y\|\leq\|\A(x)-Y\|=\|E\|,
  \end{equation*}
  and therefore we have:
  \begin{equation}\begin{split}{}
    \smallnorm{x^\sharp-x}_p	
    &\leq	\frac{2C(\kappa\rho)\kappa}{s^{1-1/p}}\sigma_s(x)_1 \\
    &+  \frac{2D(\kappa\rho)}{s^{1/q-1/p}}\Bigbr{\kappa\tau + \frac{\norm{W^{-1}}_o \norm{T}\dualnormsymbol}{s^{1-1/q}}}\norm{E}
      \label{eq:main:unscaled}
      \end{split}
  \end{equation}
  Note that $T$ can be rescaled by any positive factor, the $\mathcal{M}^+$--criterion is still fulfilled and the terms $\kappa$ and $\norm{W^{-1}}_o \norm{T}\dualnormsymbol$ in Theorem \ref{thm:M+} above do not change. However, replacing $T$ by 
   $\norm{\A^*(T)}_\infty^{-1}\cdot T$, which yields $\norm{W}_o=1$ and therefore $\kappa = \smallnorm{W^{-1}}_o$, allows us to write \eqref{eq:main:unscaled} in the more convenient form \eqref{eq:thm:main:nonneg}.
\end{proof}

\section{The Rank-one and Sub-Gaussian Case}
\label{sec:nonneg:subgaussian}
In this part we will proof our second main result, Theorem
\ref{thm:main:subgaussian}. We will consider a random linear map
$\A:\R^N\rightarrow\C^{n\times n}$ given by Model \ref{model:aa},
i.e., with the special form
\begin{align}
  \A(x)=\sum_{i=1}^N x_i a_i a_i^\herm=:\sum_{i=1}^N x_i A_i
  \label{eq:defA:rankone}
\end{align}
where $A_i \coloneqq a_i a_i^\herm \in \C^{n\times n}$ are independent
random positive-semidefinite rankone matrices. The adjoint map $\A^*: \C^{n\times n}\rightarrow\R^N$ is given as
\begin{align*}
  \A^*(T)=(\scp{A_i}{T})_{i=1}^N	=	( a^{\herm}_i T a_i)_{i=1}^N.
\end{align*}
Note that, even in the more general case where $A_i$ are (non-zero)
positive semidefinite matrices, it can been easily seen that $\A$
fulfills the $\mathcal{M}^+$ criterion since for $T=\Id_n$ we get that
$\A^*(T)=(\trace A_i)_{i=1}^N>0$.

Now, according to Model \ref{model:aa},
$\{a_i\}_{i=1}^N$ are independent complex random vectors with
independent subgaussian real and imaginary part components.  To make
this precise we will need the following characterization.  For a
real-valued random variable $X$ and $r\in [1,\infty)$ define
\begin{equation}
  \norm{X}_{\psi_r}	\coloneqq	\inf\set{t>0 : \expec{\exp(\abs{X}^r/t^r)} \leq 2}.
   \label{eq:def:psip}
\end{equation}
This is a norm on the \textit{Orlicz space} of random variables $X$
with $\normpsi{X}{r}<\infty$. For $r=2$ these random variables are called \textit{sub-Gaussian} and for $r=1$ \textit{sub-exponential}. More information about these spaces can be found in \cite{Foucart2013} and \cite{Vershynin:datasciencebook} for example. 

\subsection{The ${\cal M}^+$--Criterion}
We already discussed above that the measurement map $\A$ in
\eqref{eq:defA:rankone} fulfills the $\mathcal{M}^+$-criterion (by
choosing $T=\Id_n$ or a scaled version).  However, its ``quality'' depends (for a chosen
$T$) on the condition number $\kappa$ which is a random variable.
We follow
the ideas of \cite{Kueng2018} again.

\begin{lem}\label{lem:M+AM}
  Assume that $\A$ is given by Model \ref{model:aa}.
  For a given $\eta\in(0,1)$ it holds with high probability at least
  \begin{equation*}
    1 - 2N\exp \bigbr{-\frac{c\eta^2}{2\psi_2^4}\cdot n },
  \end{equation*} 
  that for all $i\in[N]$
  \begin{align}\label{eq:M+bound1}
    & n(1-\eta)\leq  \norm{a_i}_2^2 \leq n(1+\eta),
  \end{align}
   where $c>0$ is the constant appearing in the Hanson-Wright inequality \eqref{eq:Hanson-Wright}.
  In particular,
  \begin{align}\label{eq:M+bound2}
    &\frac{\max_{i\in [N]} \norm{a_i}_2^2}{\min_{i\in [N]} \norm{a_i}_2^2}	\leq 	\frac{1 + \eta}{1 - \eta}.
  \end{align}
\end{lem}
A variant of Lemma \ref{lem:M+AM} is possible for random vectors
beyond the iid. model if a convex concentration property hold, see
\cite{Adamczak2015}. Let us already indicate how we will use this
result later on. In the context of proving Theorem
\ref{thm:main:subgaussian} applied to Model \ref{model:aa} with
$T=t\Id_n$ we have
$\kappa = \frac{\max_{i\in [N]} \norm{a_i}_2^2}{\min_{i\in [N]}
  \norm{a_i}_2^2}$ and
$\norm{\A^*(T)}_\infty=t\max_{i\in [N]} \norm{a_i}_2^2$.  Thus, Lemma
\ref{lem:M+AM} allows us to control the terms related to
the $\mathcal{M}^+$-criterion. We will do this more explicitely below
when proving Theorem \ref{thm:main:subgaussian}.
\begin{proof}
  Note that $\expec{\|a_i\|_2^2}=n$.  We will show that
  with high probability it holds for all
  $i\in[N]$ that 
  \begin{equation*}
    |\|a_i\|_2^2 - n|	\leq	 \eta n.
  \end{equation*}
  This directly implies \eqref{eq:M+bound1} and \eqref{eq:M+bound2}.
  Using the Hanson-Wright inequality \eqref{eq:Hanson-Wright}
  (which is a Bernstein inequality in this case) yields that for all $i\ind$ it holds that
  \begin{align*}
    &\prob{\smallabs{ \norm{a_i}_2^2 - n} \geq n\eta}\\
    \leq&\,	2 \exp \big( -c n \min \{ \frac{\eta^2}{2\psi_2^4}, \frac{\eta}{\psi_2^2}\}\big)\\
    =& \,2\exp \bigbr{-\frac{c\eta^2}{2\psi_2^4}\cdot n },
  \end{align*}
  using $\psi_2\geq 1$ and $\eta<1$.
  As a remark, such a concentration may also hold for certain
  non-iid. models, see here the convex concentration property
  \cite{Adamczak2015}.  By taking the union bound it follows that
  \eqref{eq:M+bound1} and \eqref{eq:M+bound2}  hold with probability
  \begin{align*}
    &\geq 1 - 2N\exp \bigbr{-\frac{c\eta^2}{2\psi_2^4}\cdot n },
  \end{align*}
  depending on $\psi_2$, the dimensions $n$ and $N$ and some
  $\eta\in (0,1)$.
\end{proof}
\subsection{The Nullspace Property}
We will now establish that the $\ell^2$--NSP holds
with overwhelming probability once the sparsity $s$ is below a certain
threshold, in detail $s\lesssim m/\log^2(N/m)$ where $m=2n(n-1)$.
This resembles that the well-known compressed sensing phase
transition holds (up the order of the logarithm) also for such structured
random matrices.

It is well-known that the $\ell^2$--NSP is implied by the \textit{restricted isometry property}
(with respect to the $\ell^2$--norm).
\begin{defi}
  For $s\ind$, the \textit{restricted isometry constant}
  $\delta_s=\delta_s(\Phi)$ of order $s$  of a matrix
  $\Phi\in\C^{m\times N}$ is defined as the
  smallest $\delta\geq0$ satisfying
  \begin{equation}
    (1-\delta)\norm{x}_2^2	\leq	\norm{\Phi x}_2^2	\leq	(1+\delta)\norm{x}_2^2
  \end{equation}
  for all $s$-sparse vectors $x\in\R^N$, i.e. vectors with at most $s$
  non-zero components.  If $\delta_s(\Phi) <1$, the matrix $\Phi$ is
  said to have the {\em restricted isometry property} ($\ell^2$--RIP) of order $s$.
\end{defi}
If
$\delta_{2s}(\Phi)<1/\sqrt{2}$, then $s$-sparse vectors $x\in\R^N$ can be
recovered in a stable way from given measurements $\Phi x$ using
$\ell^1$-based convex algorithm (basis pursuit etc.) \cite{Cai2014}.  The following
theorem \cite[Theorem 6.13]{Foucart2013} shows the important
relation to the nullspace property.
\begin{thm}
  \label{thm:RIPNSP}
  If the 2sth restricted isometry constant
  $\delta_{2s}=\delta_{2s}(\Phi)$ of a matrix $\Phi\in\C^{m\times N}$
  obeys $\delta_{2s}\leq\delta <\frac{4}{\sqrt{41}}$, then $\Phi$
  satisfies the $\ell^2$--NSP of order $s$ with constants
  \begin{align}\label{eq:rho_tau_RIP}
    &\rho \leq \frac{\delta}{\sqrt{1-\delta^2}-\delta/4}\quad\text{and}\quad
      \tau \leq \frac{\sqrt{1+\delta}}{\sqrt{1-\delta^2}-\delta/4}.
  \end{align}
\end{thm}
For the proof see \cite[Theorem 6.13]{Foucart2013}.
For example, as seen in Figure \ref{fig:nspparam}, $\delta=0.5$ gives
$\rho\lessapprox 0.7$, $\tau\lessapprox 1.5$ and the constants in Theorem \ref{thm:1}
are $C(\rho)\approx 8.6$ and $D(\rho)\approx 11.3$.
\begin{figure}
  \hspace*{-1em}
  \includegraphics[width=1.1\linewidth]{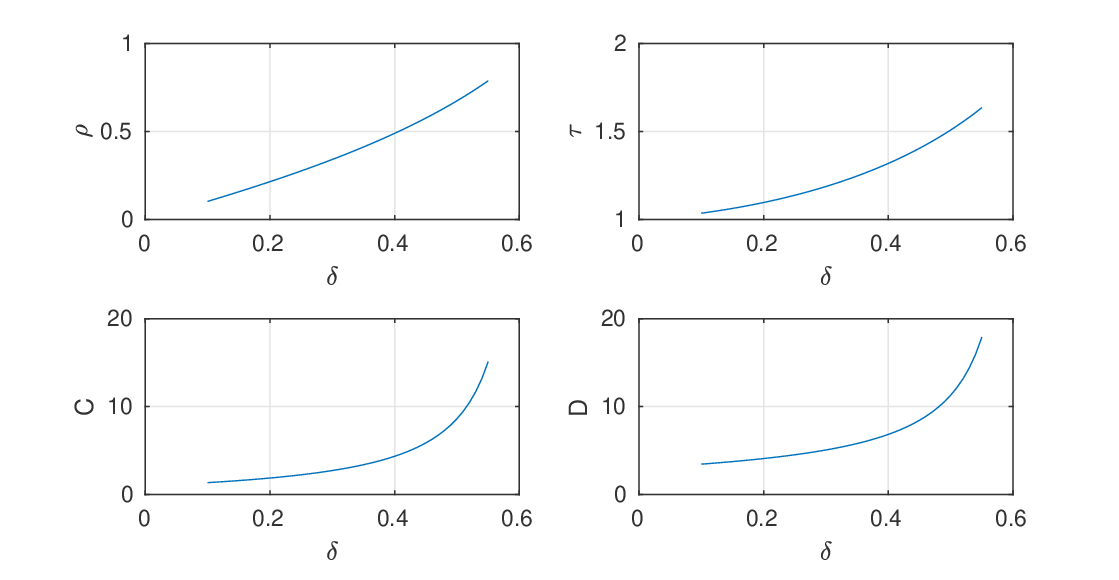}
  \caption{Dependency of the NSP parameter (bounds) $\rho=\rho(\delta)$,
    $\tau=\tau(\delta)$ in Theorem \ref{thm:RIPNSP}
    and the constants
    $C(\rho(\delta))$ and $D(\rho(\delta))$
    in Theorem \ref{thm:1}. For example, $\delta=0.5$ gives $\rho\approx 0.7$, $\tau=1.5$,
    $C(\rho)\approx 8.6$ and $D(\rho)\approx 11.3$.}
  \label{fig:nspparam}
\end{figure}  
Our first step will be to show that in the considered regime a
modified version $\Phi$ of $\A$ has with high probability $\ell^2$--RIP
with a sufficiently small RIP-constant. This then
implies that $\Phi$ and also $\A$ satisfy the $\ell^2$--NSP.

To this end, we define an operator
$P:\C^{n\times n} \rightarrow \R^{m}$, where $m\coloneqq 2n(n-1)$,
that maps a complex matrix to a real valued vector containing the real
and imaginary parts of all off-diagonal entries scaled by $\sqrt{2}$.
Hence, for any $M\in\C^{n\times n}$ we have $\|P(M)\|_2\leq\sqrt{2}\|M\|_F$.
Furthermore, we define the real vectors
\begin{equation}\begin{split}
    X_i:&=P(a_ia_i^*)\\
    &=\sqrt{2}[(\re(a_{i,k}\bar{a}_{i,l})_{k\neq l},\im(a_{i,k}\bar{a}_{i,l})_{k\neq l}]
  \end{split}
  \label{eq:def:P}
\end{equation}
These are independent and have subexponential
zero-mean entries. The factor $\sqrt{2}$ normalizes the resulting
vector so that
\begin{equation*}
  \mathbb{E}\|X_i\|_2^2
  =2\mathbb{E}\sum_{k\neq             l}\abs{a_{i,k}\bar{a}_{i,l}}^2=m.
\end{equation*}
To show $\ell^2$--RIP, we will use the following result on matrices with independent heavy-tailed columns from \cite{Adamczak2011}:
\begin{thm}[Theorem 3.3 in \cite{Adamczak2011} for the $\psi_1$-case]
  Let $X_1,\ldots,X_N \in \mathbb{R}^m$ be independent subexponential
  random vectors with $\mathbb{E} \|X_i\|_2^2 = m$ and let
  $\psi = \max_{i\in[N]}\|X_i\|_{\psi_1}$. Assume $s\leq \min \{N,m\}$ and let
  $\theta \in(0,1)$, $K,K^\prime \geq 1$ and set
  $\xi = \psi K + K^\prime$. Then for $\Phi := (X_1|...|X_N)$ it holds that
  \begin{equation}
    \delta_s\left(\frac{\Phi}{\sqrt{m}}\right) \leq C \xi^2\sqrt{\frac{s}{m}}
    \log\left(\frac{eN}{s\sqrt{\frac{s}{m}}}\right)
    + \theta
  \end{equation}
  with probability larger than
  \begin{equation}\begin{split}
      1 &- \exp(-\hat{c}K\sqrt{s}\log(\frac{eN}{s\sqrt{\frac{s}{m}}})) \\
      &- \prob{\max_{i\in [N]}\|X_i\|_2\geq K^\prime\sqrt{m}} \\
      &- \prob{\max_{i\in [N]}|\frac{\|X_i\|_2^2}{m} - 1|\geq \theta},
      \label{eq:adam_error}
    \end{split}
  \end{equation}
  where $C,\hat{c} > 0$ are universal constants.
  \label{thm:RIP:heavy}
\end{thm}
The last term in \eqref{eq:adam_error} shows the intuitive behavior
that the concentration of the column norms $\|X_i\|^2_2/m$ have direct
impact on the RIP (take for example $s=1$). In our case we will apply
Theorem \ref{thm:RIP:heavy} above to the vectors defined in \eqref{eq:def:P}. The norm $\|X_i\|^2_2$ is in
general a $4$th order polynomial in the $m$ real subgaussian random variables $\re (a_{i,k}), \im (a_{i,k})$.  
In Appendix \ref{app:psir:moments} we show how to calculate tail bounds for a polynomial of this form, the summary for our specific case is the following corrollary: 

\begin{cor}\label{cor:tailPaa}
  Consider the model \ref{model:aa} and $X_i\coloneqq P(a_i a_i^*)$ as
  defined in \eqref{eq:def:P} so that $\mathbb{E}(\|X_i\|_2^2)=m$. Assume $n\geq \psi_2^4$.
  For $\omega\in [0,1]$ it holds that
    \begin{align*}
        \prob{ \abs{\|X_i\|_2^2 - m} \geq m\omega} \leq 2 \exp\bigbr{ -\gamma \frac{\omega^2}{{\psi_2^4}} \cdot n  },
    \end{align*}
    where $\gamma\in (0,1)$ is some absolute constant.
    
    
   \end{cor}
  \begin{proof}
    This follows from Proposition \ref{prop:conc:4ord} in Appendix
    \ref{app:psir:moments}. In our case we have
    $\mu=0$ and $\sigma^2={1}$ and $L=\psi_2\geq 1$, hence the
    minimum in \eqref{eq:minZeta} can be computed as
    $\min\{\frac{\omega^2}{\psi_2^8}\cdot n,\frac{\omega^2}{\psi_2^4}\}=\frac{\omega^2}{\psi_2^4}$, using $n\geq \psi_2^4$.
  \end{proof}

  Now we are ready to show $\ell^2$--RIP for the matrix
  $\Phi \coloneqq \frac{1}{\sqrt{m}}P\circ\A$. A similar result for
  the real case where (informally) ``$P$ is replaced by centering''
  has been established in \cite{Fengler:krrip:2019}.  However, to
  establish the NSP it is more direct to remove the diagonal part with
  the definition of $P$ in \eqref{eq:def:P}.
\begin{thm}\label{thm:RIP}
  Assume that $\A:\R^N\rightarrow\C^{n\times n}$ is given by Model
  \ref{model:aa}
  and let $\delta\in(0,1]$. Assume
  $N\geq m=2n(n-1)$. If
  \begin{equation}
    2s\leq \alpha m\log^{-2}(\frac{eN}{\alpha m})
    \label{eq:thm:RIP}
  \end{equation}
  and $n\geq \frac{2 \log(4N)}{C_1}$,  then,
  with probability
  \begin{equation}
    \geq 1-2\exp(- \min\smallset{\hat{c}  \sqrt{\alpha},\, \frac{1}{2}C_1} \cdot n ),\label{eq:RIPprob}
  \end{equation}
  the matrix $\Phi=\frac{1}{\sqrt{m}}P\circ \A\in\R^{m\times N}$
  has $\ell^2$--RIP of order $2s$ with RIP-constant $\delta_{2s}(\Phi) \leq \delta$. 
  The constants $C,\hat{c}$ are the same as in Theorem \ref{thm:RIP:heavy} and  $C_1, \alpha$ are given as  $C_1= \frac{\gamma  \delta^2}{4\psi_2^4}$, with $\gamma$ as in Corollary \ref{cor:tailPaa}, and $\alpha\coloneqq\min\{1,
  \big(\frac{\delta}{6C(\psi_1+\sqrt{1+\delta/2})^2}\big)^2\}$, where $\psi_1\coloneqq \max_{i\in [N]}\norm{P(a_i a_i^*)}_{\psi_1}$.
\end{thm}
\begin{proof}
  We will apply Theorem  \ref{thm:RIP:heavy} and use ideas already
  presented in \cite[Theorem 5 and Corollary 1]{Fengler:krrip:2019}.
  
  Define the $N$ real-valued random vectors
  $X_i=P(a_ia_i^*), i\in [N]$. The number
  $\psi_1=\max_{i\in[N]}\normpsi{X_i}{1}$ defined above is finite,
  independent of the dimension and depends quadratically on $\psi_2$, see Appendix
  \ref{app:psir:moments}.  Let $\alpha\in (0,1]$, the value will be
  specified later, and set
  $s^* \coloneqq \alpha m/\log^2(\frac{eN}{\alpha m})$. Since
  $\log(\frac{eN}{\alpha m})\geq 1$, we ensure $s^*\leq m\leq N$. By
  Theorem \ref{thm:RIP:heavy}, the RIP-constant
  $\delta_{s^*}\coloneqq \delta_{s^*}(\Phi)$ of the matrix
  $\Phi\coloneqq\frac{1}{\sqrt{m}}P\circ \A$ satisfies
  \begin{equation}
    \delta_{s^*}	\leq C \xi^2 \sqrt{\frac{s^*}{m}} \log(\frac{eN}{s^* \sqrt{(s^*/m)}}) + \theta	 \label{deltaStar}
  \end{equation}
  with probability larger than
  \begin{align}
    1	&-\exp\bigbr{-\hat{c} K \sqrt{s^*} \log (\frac{eN}{s^* \sqrt{s^*/m}})}	\label{1stprob}\\
        &-	\prob{\max_{i\ind} \smallnorm{X_i}_2 \geq K' \sqrt{m}} \label{2ndprob}\\
 	&-\prob{\max_{i\ind} \smallabs{ \frac{\norm{X_i}_2^2}{m} - 1} \geq \theta}.\label{3rdprob}
  \end{align}
By definition of $s^*$, we can estimate \eqref{deltaStar} as
  \begin{align}
     \delta_{s^*}	&\leq	 \frac{C \xi^2\sqrt{\alpha}}{\log(\frac{eN}{\alpha m})} \log\Big((\frac{eN}{\alpha m})^{3/2} \log^3(\frac{eN}{\alpha m})\Big) + \theta\nonumber	\\	
                &= C\xi^2 \sqrt{c} \Big( \frac{3}{2} + 3\frac{\log \log(\frac{eN}{\alpha m})}{\log(\frac{eN}{\alpha m})}\Big) + \theta\nonumber\\
                &\leq	C\xi^2 \sqrt{\alpha } \big( \frac{3}{2} + \frac{3}{e}\big)+ \theta\nonumber\\
                &\leq 3C\xi^2\sqrt{\alpha }+ \theta,\label{eq:deltaEstimate}
  \end{align}
  where we used $\frac{\log \log x}{\log x}\leq \frac{1}{e}$ for $x>1$ in the last line.
  For \eqref{2ndprob}, \eqref{3rdprob}, taking union
  bounds and rewriting gives
  \begin{align}
    &\prob{\max_{i\ind} \smallabs{ \frac{\norm{X_i}_2^2}{m} - 1} \geq \theta}\nonumber	\\
    &\leq	N\cdot\prob{\smallabs{ \norm{X_i}_2^2 -  m} \geq \theta m}\label{eq:tailTheta}
  \end{align}
  and
  \begin{align*}
    &  \prob{\max_{i\ind} \smallnorm{X_i}_2 \geq K' \sqrt{m}}\\
    &\leq	N\cdot\prob{\smallnorm{X_i}^2_2 \geq K'^2 m}\\
    &\leq N\cdot\prob{\smallabs{ \norm{X_i}_2^2 -  m} \geq (K'^2-1) m}.
  \end{align*}  
  Choosing $K'\coloneqq \sqrt{1 + \theta}$, both terms above are
  equal. We set $\theta= \frac{\delta}{2}$.
  Note that $n\geq\frac{2\log(4N)}{C_1}$ yields $n\geq \psi_2^4$ since $C_1= \frac{\gamma  \delta^2}{4\psi_2^4}$ and $\gamma,\delta\leq 1$. Hence, using Corollary \ref{cor:tailPaa} with $\omega=\frac{\delta}{2}$, the probabilities above can be bounded by
  $2N \exp(-C_1\cdot n)$. Since
    $n\geq\frac{2 \log (4N)}{C_1}$, we can estimate
  \begin{align}
    4Ne^{-C_1\cdot n} =e^{\log(4 N) - C_1\cdot n}\leq e^{- 
    \frac{1}{2}C_1\cdot n}.
    \label{eq:unionbound:1}
  \end{align}
  Now set $K=1$ and choose $\alpha $ sufficiently
  small so that we get $\delta_{s^*}\leq \delta$ from
  \eqref{eq:deltaEstimate}, i.e., $\alpha  \leq	\big(\frac{\delta}{6C(\psi+\sqrt{1+\delta/2})^2}\big)^2$. The term \eqref{1stprob} can be estimated
  in the following way using $s^*=\alpha m/\log^2(\frac{eN}{\alpha m})\leq \alpha ^{2/3}m$
  :
  \begin{align}
    &\exp\bigbr{-\hat{\alpha } \sqrt{s^*} \log (\frac{eN}{s^* \sqrt{s^*/m}})}\nonumber\\
    \leq
    &	\exp\bigbr{	-\hat{c }	 \sqrt{s^*} \log ( \frac{eN}{\alpha m})	}	\nonumber\\	
    =&\exp\bigbr{	-\hat{c}	  \sqrt{\alpha }\cdot \sqrt{m} }
      \nonumber\\
    \leq & \exp\bigbr{-\hat{c}\sqrt{\alpha }\cdot n} \label{eq:unionbound:2}
  \end{align}
 Using \eqref{eq:unionbound:1}, \eqref{eq:unionbound:2} we get
 \begin{align*}
   \prob{\delta_{s^*} \leq \delta}&\geq
   1 - \exp\bigbr{- \hat{c}  \sqrt{\alpha } \cdot n} -
   \exp\bigbr{-\frac{1}{2} C_1 \cdot n}\\
   &\geq    1 - 2\exp\bigbr{-\min\smallset{\hat{c}\sqrt{\alpha },\frac{1}{2}C_1}\cdot n}.
 \end{align*}
 By monotonicity of the RIP-constant  we get the same lower bound for $\prob{\delta_{2s} \leq \delta}$, whenever $2s\leq s^*$.

\end{proof}
From this it easily follows that $\Phi$ and also $\A$ itself satisfy the $\ell^2$--NSP. 
\begin{thm}\label{thm:NSP}
  Assume that $\A:\R^N\rightarrow\C^{n\times n}$ is given by Model
  \ref{model:aa}.
  Let $N\geq m=2n(n-1)$, $\delta\in (0, \frac{4}{\sqrt{41}})$ and
  assume
  \begin{equation*}
    s\lesssim m\log^{-2}(N/m)
  \end{equation*}
  and $n\gtrsim \log(N)$ as in \eqref{eq:thm:RIP}.
  Then, with probability
  \begin{equation}\label{eq:NSPprob}
    \geq {1-2\exp(-c_\delta\cdot n )},
  \end{equation}
  $\A$ has the $\ell^2$--NSP of order $s$  w.r.t. the
  Frobenius norm $\|\cdot\|_\fro$ with parameters $\rho$ and $\tau\sqrt{2}/\sqrt{m}$. The number ${c_\delta}$ is defined so that \eqref{eq:NSPprob} coincides with \eqref{eq:RIPprob} and $\rho, \tau$ satisfy \eqref{eq:rho_tau_RIP} with the chosen $\delta$.
\end{thm}
\begin{proof}
  We set
  $\Phi=\frac{1}{\sqrt{m}}P\circ \A\in\R^{m\times N} $. By Theorem \ref{thm:RIP}, with probability \eqref{eq:NSPprob} $\Phi$ has $\ell^2$--RIP of order $2s$ with RIP-constant $\delta_{2s}(\Phi)\leq\delta$. Theorem \ref{thm:RIPNSP}
  implies that $\Phi$ in this case satisfies the
  $\ell^2$--NSP with parameters $(\rho,\tau)$ depending on $\delta$ as
  given in \eqref{eq:rho_tau_RIP}. Hence,
  for all $v\in\R^{N}$ and $S\subset [N]$ with
  $|S| \leq s$ it holds that
  \begin{align*}
    \norm{v_S}_2
    &\leq	\frac{\rho}{\sqrt{s}} \norm{v_{S^c}}_1 + \tau
      \norm{\Phi v}_2 \\
    &\leq	\frac{\rho}{\sqrt{s}} \norm{v_{S^c}}_1 + \frac{\tau}{\sqrt{m}} \norm{P(\A(v))}_2 \\
    &\leq
      \frac{\rho}{\sqrt{s}} \norm{v_{S^c}}_1+ \frac{\tau\sqrt{2}}{\sqrt{m}} \norm{\A(v)}_\fro,
  \end{align*}
  showing that the linear map $\A$ has the $\ell^2$--NSP of order $s$
  with respect to $\|\cdot\|_\fro$ and with parameters
  $(\rho,\tau\sqrt{2}/\sqrt{m})$.
\end{proof}

\subsection{Proof of the Main Recovery Guarantee for Model \ref{model:aa}}
Now we are ready to proceed with the proof of the second main result,
Theorem \ref{thm:main:subgaussian}. 
\begin{proof}[Proof of Main Theorem  \ref{thm:main:subgaussian}]
  We start from our first main result, Theorem \ref{thm:main:nonneg},
  for the case of the Frobenius norm $\|\cdot\|_\fro$. The convex
  program \eqref{eq:generic:nnnorm} is then {\em Nonnegative
    Least-Squares} (NNLS) and Theorem \ref{thm:main:nonneg} states
  that if the linear map $\A$ has the $\ell^2$--NSP with respect to
  $\|\cdot\|_\fro$ and fulfills the ${\cal M}^+$-criterion for some
  matrix $T$ with a sufficiently well-conditioned
  $\kappa=\kappa(\A^*(T))$, then NNLS obeys a recovery guarantee of
  the form \eqref{eq:main:unscaled}. It will be more convenient to
  choose here a different scaling for $T$ as we did in the end of the
  proof of Theorem \ref{thm:main:nonneg}.
  
  Theorem \ref{thm:NSP} states that with high probability $\A$ has the
  $\ell^2$--NSP with parameters $(\rho, \sqrt{2} \tau/\sqrt{m})$,
  where $\rho,\tau$ depend on the number $\delta$ from Theorem
  \ref{thm:RIPNSP} and \ref{thm:RIP}. We know that the
  $\mathcal{M}^+$--criterion for $\A$ is fulfilled for
  $T=t\cdot \Id_n$ with $t>0$. Lemma \ref{lem:M+AM} furthermore states
  that with overwhelming probability the resulting vector
  $w=t\A^*(\Id_n)$ is well-conditioned and concentrates around its
  mean. Set $\kappa \coloneqq \kappa(w)$ and $W\coloneqq \diag (w)$.
  Conditioned on events
  when $\A$ indeed has the $\ell^2$--NSP and $\kappa\rho<1$, we have from
  \eqref{eq:main:unscaled} that for any $1\leq p\leq q=2$ it holds that
  \begin{equation}\begin{split}\label{eq:main:subgaussian:1}
      \|&x^\sharp-x\|_p\\
      &\leq	\frac{2C(\kappa\rho)\kappa}{s^{1-\frac{1}{p}}}\sigma_s(x)_1\\
      &+  \frac{2D(\kappa\rho)}{s^{\frac{1}{2}-\frac{1}{p}}}\Bigbr{\kappa\frac{\sqrt{2}\tau}{\sqrt{m}} +
          \frac{\norm{W^{-1}}_o\norm{T}\dualnormsymbol}{\sqrt{s}}}\norm{E}_\fro.
      \end{split}
    \end{equation}
 The equation \eqref{eq:M+bound2} in this setting translates to $\kappa(w)\leq
   \frac{1+\eta}{1-\eta}=:\kappa_\eta$, where $\eta\in (0,1)$ will be specified later. Recall that the condition number is invariant to scaling of $w$, hence $\kappa=\kappa(w)$. The dual norm in
    \eqref{eq:thm:main:nonneg} is
    $\|T\|^\circ=\|T\|_\fro=t\|\Id_n\|_\fro=t\sqrt{n}$ and $\norm{W^{-1}}_o = (t\min_{i\in[N]} \norm{a_i}_2^2)^{-1}\leq \bigbr{tn(1-\eta)}^{-1}$. Choosing $t\coloneqq \bigbr{n(1+\eta)}^{-1}$ we achieve $\norm{W^{-1}}_o \leq \kappa_\eta$ and $\norm{T}\dualnormsymbol =\bigbr{\sqrt{n}(1+\eta)}^{-1} $. With these bounds and setting
  $C_{\eta,\rho}=2C(\kappa_\eta\rho)\kappa_\eta$,
  $D_{\eta,\rho}=2D(\kappa_\eta\rho)\kappa_\eta$, we can further estimate \eqref{eq:main:subgaussian:1}  as
  \begin{equation}\begin{split}\label{eq:main:subgaussian:2}
   &\leq	\frac{C_{\eta,\rho}\sigma_s(x)_1}{s^{1-\frac{1}{p}}}
      +  \frac{D_{\eta,\rho}}{s^{\frac{1}{2}-\frac{1}{p}}}\Bigbr{\frac{n \sqrt{2}\tau}{\sqrt{m}} +  \sqrt{\frac{n}{s}}(1+\eta)^{-1}}\frac{\norm{E}_\fro}{n}\\
  &\leq	\frac{C_{\eta,\rho}\sigma_s(x)_1}{s^{1-\frac{1}{p}}}
      +  \frac{D_{\eta,\rho}}{s^{\frac{1}{2}-\frac{1}{p}}}\Bigbr{2\tau +
          \sqrt{\frac{n}{s}}(1+\eta)^{-1}}\frac{\norm{E}_\fro}{n},
    \end{split}
  \end{equation}
  In particular the last step may be improved further by explicitly accounting
  for the bound in \eqref{eq:thm:RIP}. Instead we have assumed only $n>1$
  so that
  $\frac{n}{\sqrt{m}}=\frac{n}{{\sqrt{2n(n-1)}}}\leq \sqrt{2}$.
  
  A possible concrete choice of the not yet specified numbers is
  $\eta=1/3$ and $\delta = 1/6$, see here also Figure
  \ref{fig:nspparam1}.  In this case we have $\kappa_{\eta} = 2$ and
  $\rho\leq 0.18$, hence $\kappa\rho < 1$ is fulfilled,
  $\tau\leq 1.15 $ and
  $C_{\eta,\rho}\leq 11.36, D_{\eta,\rho}\leq 20.73$.
Plugging into \eqref{eq:main:subgaussian:2} yields the desired inequality \eqref{eq:thm:main:subgaussian}
 \begin{align*}
       \smallnorm{x^\sharp-x}_p
       \leq	\frac{c_2\sigma_s(x)_1}{s^{1-\frac{1}{p}}}
       +  \frac{c_3\br{c_4 +
           \sqrt{\frac{n}{s}}}}{s^{\frac{1}{2}-\frac{1}{p}}}\frac{\norm{E}_\fro}{n}
  \end{align*}
  with constants 
  \begin{align*}
    c_2&=C_{\eta,\rho}\leq 11.36,\\ c_3&=D_{\eta,\rho}(1+\eta)^{-1}\leq 15.55,\\ c_4&=2\tau(1+\eta)\leq 3.07.
  \end{align*}
  The probability for \eqref{eq:M+bound1},  \eqref{eq:M+bound2} to hold can be estimated as
  \begin{align*}
      &1 - 2N\exp \bigbr{-\frac{c}{18\psi_2^4}\cdot n }\\
      \geq & 1  - 2\exp\bigbr{-\frac{c}{36\psi_2^4}\cdot n}
  \end{align*}
  if $n\geq \frac{36\psi_2^4 \log(N)}{c}$. Taking a union bound with \eqref{eq:NSPprob} gives a probability of at least $1 - 4 \exp\bigbr{-c_1\cdot n }$ with 
  $c_1\coloneqq\min \{ \frac{c}{36\psi_2^4}, \hat{c}\sqrt{\alpha}, \frac{1}{2}C_1 \}$
  for \eqref{eq:thm:main:subgaussian} to hold if also $n\geq\frac{2\log(4N)}{C_1}$, where $c$ is the constant from the Hanson-Wright inequality and $\hat{c}, \alpha, C_1$ are the same as in Theorem \ref{thm:RIP} and depend on $\psi_2$ but not on the dimensions.
  \begin{figure}
    \hspace*{-1em}
    \includegraphics[width=1.1\linewidth]{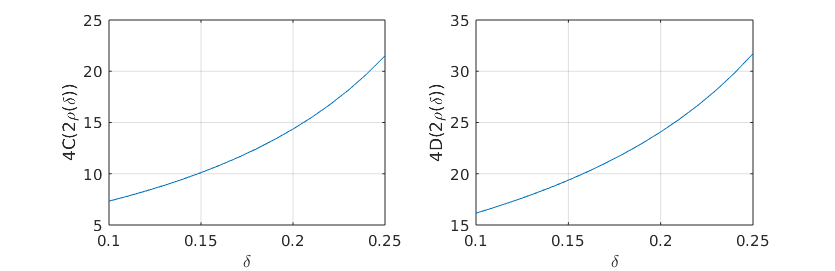}
    \caption{Dependency of the constants
      $C_{\eta,\rho}$ and $D_{\eta,\rho}$
      in \eqref{eq:main:subgaussian:2}
      depending on $\delta$ for fixed $\eta=1/3$ (yielding
      $\kappa_\eta=2$ and therefore $C_{\eta,\rho}=4C(2\rho(\delta)))$
      and $D_{\eta,\rho}=4D(2\rho(\delta)))$, where $C$ and $D$ are
      defined as in Theorem \ref{thm:1} and shown in Figure \ref{fig:nspparam}).}
    \label{fig:nspparam1}
  \end{figure}

\end{proof}

\section{Numerical Experiments}
In the following we validate our theoretical result
\eqref{eq:thm:main:subgaussian:phasetransition} in Theorem
\ref{thm:main:subgaussian} about the phase transition for successful
recovery via NNLS for Model \ref{model:aa} with numerical experiments.
We performed recovery experiments for dimensions $n=20,\dots, 30$ and
sparsity range $s=20,\dots, 150$. For every pair $(n,s)$ we have performed
$20$ experiments with randomly generated vectors $\{a_i\}_{i=1}^N$
with independent standard normal entries and a nonnegative sparse
vector $x\in\R^N$. The support of $x$ is generated uniformly over all
possible $\binom{N}{s}$ combinations. The nonnegative values on the support are
generated independently as absolute values from a standard normal
distribution.  Given the noiseless measurement $Y=\A(x)$, we then used
the MATLAB function \texttt{lsqnonneg} to solve the NNLS problem (the
convex program \eqref{eq:generic:nnnorm} for the Frobenius norm)
yielding the estimate $x^\sharp$. We assume that the vector is
successfully recovered if $\|x-x^\sharp\|_2\leq 10^{-4}$. The
corresponding result is shown in Figure \ref{fig:phasetrans}.

\begin{figure}[h]
  \includegraphics[width=1\linewidth]{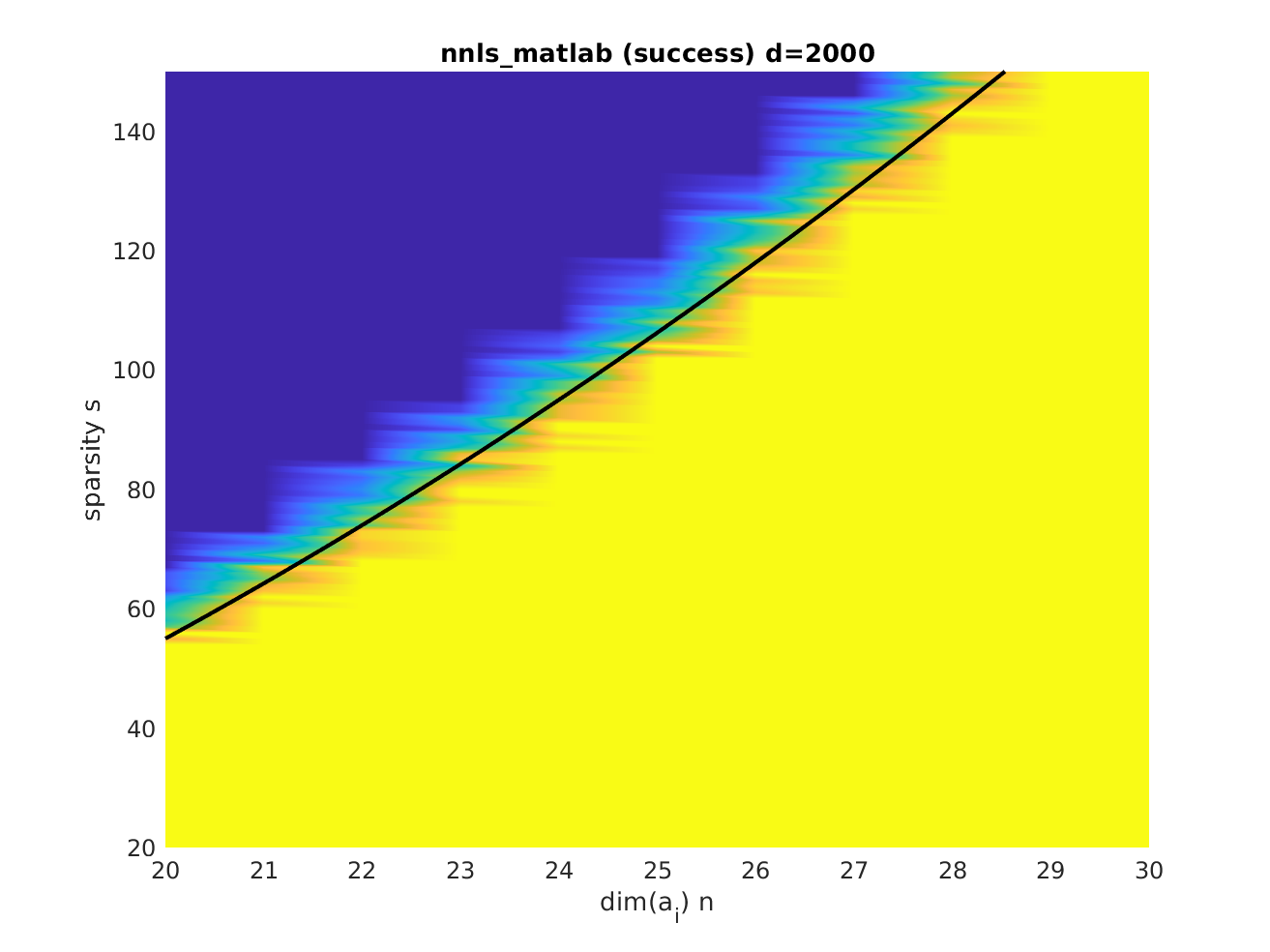}
    \caption{Phase transition for NNLS (the convex program
    \eqref{eq:generic:nnnorm} for the Frobenius norm) in the noiseless
    case (success=light/yellow and failure=blue/dark).
    The function $x\rightarrow x^2/4-x-25$ is overlayed in
    black.}
  \label{fig:phasetrans}
\end{figure}

\section*{Acknowledgments}
We thank Alexander Fengler, Radoslaw Adamczak and Saeid
Haghighatshoar.
PJ has been supported by DFG grant JU 2795/3.
The work was partially supported by DAAD grant 57417688.

\begin{appendices}
  
  \section{Hanson Wright Inequality}
  The Hanson-Wright inequality is an important tool to calculate tail bounds for sub-Gaussian random vectors. We first state it for the real case, taken from \cite[Theorem 1.1]{Rudelson2013}.
  \begin{thm}[Hanson-Wright inequality]
    Let $a = (a_1,\ldots,a_n)\in\R^n$ be a random vector with
    independent and centered
    sub-Gaussian components and $Z\in\R^{n \times n}$. For all $t\geq 0$ it holds that
    \begin{equation}\begin{split}
      \label{eq:Hanson-Wright}
      &\prob{\abs{\langle a, Z a\rangle - \expec{\langle a, Z a\rangle}} > t} \\
      &\leq 2 \exp(-c \min\{\frac{t^2}{K^4 \norm{Z}_F^2}, \frac{t}{K^2 \norm{Z}_o}\}),
      \end{split}
    \end{equation}
    where $K$ is a bound on the $\psi_2$-norms of the components of $a$ and $c>0$ a universal constant.
  \end{thm}
  The complexifications have been discussed in \cite[Sec. 3.1]{Rudelson2013}.
  One important application for us is bounding the deviation of the Euclidian norm squared of a complex vector $a\in\C^n$ from its mean by writing   
  \begin{align*}
      \norm{a}_2^2= \norm{\tilde{a}}_2^2 = \scp{\tilde{a}}{I_{2n}\tilde{a}},
  \end{align*}
  where $\tilde{a}\coloneqq\begin{bmatrix}\re(a)\\ \im(a)\end{bmatrix}\in\R^{2n}$ and $I_{2n}$ is the $2n\times 2n$ identity matrix with $\norm{I_{2n}}_F^2=2n$ and $\norm{I_{2n}}_o=1$. 
  But we can furthermore even state a complete complex version.

  \begin{thm}[Hanson-Wright inequality, complex version]
  \label{thm:HWcomplex}
    Let $a = (a_1,\ldots,a_n)\in\C^n$ be a random vector so that $\re(a_i),\im(a_i)$ are
    independent and centered
    sub-Gaussian random variables and let $Z\in\C^{n \times n}$. For all $t\geq 0$ it holds that
    \begin{equation}\begin{split}
      \label{eq:HWcomplex}
      &\prob{\abs{\langle a, Z a\rangle - \expec{\langle a, Z a\rangle}} > t} \\
      &\leq 4 \exp(-c \min\{\frac{t^2}{4K^4 \norm{Z}_F^2}, \frac{t}{\sqrt{2}K^2 \norm{Z}_o}\}),
      \end{split}
    \end{equation}
    where $K$ is a bound on the $\psi_2$-norms of the real and imaginary parts of the components of $a$ and $c>0$ the same constant as in \eqref{eq:Hanson-Wright}.
  \end{thm}
  \begin{proof}
    Taking squares on both sides and using $\abs{\cdot}^2 = \re(\cdot)^2 + \im(\cdot)^2$ yields
    \begin{align}
        &\prob{ \abs{\scp{a}{Za} - \expec{\scp{a}{Za}}}>t}\nonumber\\
         = & \prob{\re^2(\scp{a}{Za} - \expec{\scp{a}{Za}})\nonumber\\
        &+\im^2(\scp{a}{Za} - \expec{\scp{a}{Za}})  >t^2}\nonumber\\
        \leq & \prob{\abs{\re(\scp{a}{Za} - \expec{\scp{a}{Za}})} \geq \frac{1}{\sqrt{2}}t}\label{term:HWrealPart}\\
            &+\prob{ \abs{\im(\scp{a}{Za} - \expec{\scp{a}{Za}})} \geq \frac{1}{\sqrt{2}}t}.\label{term:HWimPart}
    \end{align}{}
    Writing
    \begin{align*}
        &\scp{a}{Za}\\
        =& \bigbr{ \re(a)^T - \i \im(a)^T }\bigbr{ \re(Z) + \i \im(Z) }\\
        &\cdot \bigbr{ \re(a) + \i \im(a) }\\
        =&  \begin{bmatrix} \re(a)^T & \im(a)^T  \end{bmatrix} \begin{bmatrix} \re(Z) & -\im(Z) \\ \im(Z) & \re(Z) \end{bmatrix}
      \begin{bmatrix} \re(a) \\ \im(a) \end{bmatrix}\\
      &+ \i \begin{bmatrix} \re(a)^T & \im(a)^T  \end{bmatrix} \begin{bmatrix} \im(Z) & \re(Z) \\ -\re(Z) & \im(Z) \end{bmatrix}
      \begin{bmatrix} \re(a) \\ \im(a) \end{bmatrix}\\
      \eqqcolon& \tilde{a}^T\tilde{Z}_1\tilde{a} + \i\,\tilde{a}^T\tilde{Z}_2\tilde{a} ,
    \end{align*}
    we can apply the Hanson-Wright inequality for the real case to \eqref{term:HWrealPart} and \eqref{term:HWimPart} with $\smallnorm{\tilde{Z}_{1/2}}_{HS} = \sqrt{2} \norm{Z}_{HS}$ and $\smallnorm{\tilde{Z}_{1/2}}_{o} = \norm{Z}_o$, to obtain the result.
  \end{proof}{}



  \section{Concentration of $4$th order Polynomials - Full Approach}
  \label{app:conc4}
  To calculate the probabilities of the form $\prob{\abs{\smallnorm{X_i}^2_2 - m} \geq \omega m}$ appearing in \eqref{2ndprob}, \eqref{3rdprob}, we observe that $\norm{X_i}_2^2$ is essentially a $4$th order polynomial in the sub-Gaussian random variables $\re(a_{i,k})$ and $\im(a_{i,k})$. 
   Setting $v_k \coloneqq \re(a_{i,k})$ and
    $v_{n+k} \coloneqq \im(a_{i,k})$, a quick calculation shows that we
    can write this as
    \begin{align*}
    \norm{X_i}_2^2
        &=\sum_{k,l\inn,k\neq l}(v^2_k + v^2_{n+k})(v^2_l + v^2_{n+l})\\
        &=\sum_{(k,l)\in I} v_k^2 v_l^2,
    \end{align*}
    setting 
    \begin{equation}\label{eq:setI}
        I=\{(k,l)\inzn\times[2n] : k\neq l, k\neq n+l, l\neq n+k\}
    \end{equation}{}
   The following theorem, which can be seen as a generalization of the Hanson-Wright inequality, allows to analyze these terms. 
  \newcommand{\J}{\mathcal{J}}
  \begin{thm}[Theorem 1.6 in \cite{Goetze2019}]\label{thm:concentrationPoly}
    Let $Z=(Z_1,\dots,Z_\ell)$ be a random vector with independent
    components, such that $\|Z_i\|_{\psi_2}\leq L$ for all $i\in [\ell]$. Then, for
    every polynomial $f:\R^\ell\rightarrow\R$ of degree $D$ and all $t>0$, it holds that 
    \begin{align*}
      & \prob{ |f(Z)-\mathbb{E} f(Z)|\geq t}\\
      &\leq 2\exp\big(-\frac{1}{C_D}\min_{1\leq d\leq D}\min_{\J\in
        P_d}\eta_{\J}(t)\big)
    \end{align*}
    where
    \begin{equation}
      \eta_{\J}(t)=\left(\frac{t}{L^d\|\mathbb{E}\mathbf{D}^d f(Z)\|_{\J}}\right)^{2/\#\J}.\label{eq:defEta}
    \end{equation}
  \end{thm}
  \newcommand{\D}{\mathbf{D}}
  Here $\mathbf{D}^d f$ is the $d$-th derivative of $f$ and for a
  multi-index array $W=(w_{i_1\dots i_d})_{i_1\dots i_d=1}^\ell$ the $\norm{\cdot}_\J$-norm is defined as
  \begin{equation}
    \begin{split}
      \|W\|_\J:=\sup\{&\sum_{\mathbf{i}\in[\ell]^d}w_{\mathbf{i}}\prod_{l=1}^k(x_{\mathbf{i}_{J_l}})^l\,|\,
      \|x_{\mathbf{i}_{J_l}}\|_2\leq1\\
      & \text{for all } l\in [k]\},
    \end{split}
  \end{equation}
  where $\J=(J_1,\dots,J_k)\in P_d$ is a partition of $[d]$
  into non-empty, pairwise disjoint sets. Some examples are:
  \begin{equation*}
    \begin{split}
      \|W\|_{\{1,2\}} &=\|W\|_F\\
      \|W\|_{\{1\}\{2\}} &=\|W\|_o\\
      \|W\|_{\{1,2\}\{3\}} &=\sup_{\|x\|_F\leq 1\,\&\,\|y\|_2\leq 1}\sum_{ijk} w_{ijk}x_{ij}y_j
    \end{split}
  \end{equation*}
    
    Our first calculation allows the analysis of the deviation of $\norm{Z}_2^2$ from its mean for a complex sub-Gaussian random vectors $Z$ with iid. components.
  \begin{prop}
    Let $Z=(Z_1,\dots,Z_{2n})$ be a random vector with independent
    components, such that $\|Z_i\|_{\psi_2}\leq L$,
    $|\expec{ Z_i} |\leq \mu$ and $\expec{Z_i^2}\leq \frac{1}{2} \sigma^2$ for some $L\geq 1,\, \mu,\sigma^2\geq 0$ and all
    $i\inzn$.  Consider the $4$-th order polynomial
    \begin{align*}
      f: \,\R^{2n} \rightarrow \R,\quad
          v \mapsto \sum_{(k,l)\in I} v_k^2 v_l^2 ,
    \end{align*}
    where $I$ is given as in \eqref{eq:setI}.
    Assume $n\geq 2$\note{?}.
    Then for all $\omega>0$ it holds that
    \begin{align*}
      &\prob{ |f(Z)-\mathbb{E} f(Z)|\geq n(n-1)\omega}\\
      &\leq 2\exp(-\gamma\,\zeta\cdot n)
    \end{align*}
    where $\gamma\in (0,1)$ is an absolute constant and 
    \begin{equation}\begin{split}\label{eq:minZeta}
      \zeta=&\min\{\frac{\omega^2}{L^2\mu^2\sigma^4},\frac{\omega}{L^2(\sigma^2 + 2 \mu^2)},\frac{\omega^2}{L^4 (\sigma^2 + \mu^2)^2},\\
    &\frac{\omega^{2/3}}{L^2\mu^{2/3}}, \frac{\omega}{L^3\mu},\frac{\omega^2}{L^6 \mu^2}\cdot n, \frac{\omega^{1/2}}{L^2},
     \frac{\omega^{2/3}}{L^{8/3}}, \frac{\omega}{L^4},\frac{\omega^2}{L^8}\cdot n
                    \}.
    \end{split}\end{equation}
    \label{prop:conc:4ord}
  \end{prop}
  Note that two of the terms in \eqref{eq:minZeta} contain a factor $n$ and will therefore not play a role for large $n$.
  \begin{proof}
    The partial derivatives are
    \begin{align*}
      \partial_{i} f(v)
      &=4v_i\sum_{k\inzn,(i,k)\in I} v_k^2 \\
      \partial_{i,i} f(v)
      &=	4\sum_{k\inzn, (i,k)\in I} v_k^2 ,\\
      \partial_{i,j} f(v)
      &=	8v_iv_j,\\
      \partial_{i,i,j}f(v) 
      &= 8v_j\\
      \partial_{i,i,j,j} f(v) &=8,\\
    \end{align*}
    for all $(i,j)\in I$. All combinations not mentioned here are zero or follow from the calculations above by Schwarz's theorem about mixed partial derivatives. 
    We have to estimate \eqref{eq:defEta} for all possible partitions $\J$ and $t=\omega m$. We will only state some of the calculations here, the others follow in a similar manner. Note that $\# I =2n(2n-2) $ and  for any $i\inzn$, there are  $2(n-1)$ indices $k\inzn$ such that $(i,k)\in I$.
    For the case $\J=\{1\}$,
    \begin{align*}
      &\norm{\expec{ \D^1 f(Z)}}_{\{1\}}\\
      =&\sup\{ \sum_{i\in[2n]}\expec{\partial_{i} f(Z)}x_{i}       \,|\,x\in\R^{2n} \text{ with } \norm{x}_2\leq 1\},
    \end{align*}
    let $x\in\R^{2n}$ with $\norm{x}_2\leq 1$. Since
    \begin{align*}
      &\sum_{i\in[2n]}\expec{\partial_{i} f(Z)} x_{i}    =4\sum_{(i,k)\in I}\expec{Z_i}\,\expec{Z_k^2} x_{i} \\
      \leq&	4 \mu \sigma^2 \sum_{(i,k)\in I} x_{i}
      \leq 4 \mu\sigma^2 (2n-2) \norm{x}_1 \\
      \leq& 8 \mu\sigma^2 (n-1)\sqrt{2n},
    \end{align*}
    we get the estimate
    \begin{align*}
      \eta_{\{1\}}(\omega m)	&\geq	\br{\frac{2\omega n(n-1) }{L^1 \cdot 8\mu \sigma^2 (n-1)\sqrt{2n}}}^{2/1}	\\
                                &=	\frac{\omega^2}{32L^2 \mu^2 \sigma^4 }\cdot n.
    \end{align*} 
    \newcommand{\F}{\bold{F}}
    To illustrate another important technique, consider
    $\J=\{1,2\}\{3\}$ and $x\in\R^{2n\times2n}, y\in\R^{2n}$ with
    $\|x\|_\fro=\|y\|_2=1$. We can assume $x,y\geq 0$, entrywise, to calculate the upper bound.
    \begin{align*}
      &\expec{\D^3 f(Z)}(x,y)	\\
      =& \sum_{i,j,k\inzn} \expec{\partial_{i,j,k}f(Z)}x_{ij}y_k \\
      =&\sum_{(i,j) \in I}8\expec{Z_j}x_{ii}y_j + 8\expec{Z_j}x_{ij}y_i + 8\expec{Z_j}x_{ji}y_i\\
      \leq	&8\mu \sum_{(i,j)\in I} x_{ii} y_j+ x_{ij}y_i + x_{ji} y_i\\
      \leq&	8 \mu \big( \smallnorm{\diag(x)}_1 \smallnorm{y}_1	+ \sum_{j\inzn} (\smallscp{x_j}{y} +  \smallscp{{}_jx}{y}) \big)\\
      \leq	& 	8\mu \big( \sqrt{2n}\cdot\sqrt{2n}		+ \sum_{j\inzn} (1+ 1)\big)\\
      =	&48\mu n,
    \end{align*}
    where by $x_j, {}_jx$ we denoted the $j$-th row respectively column of $x$ and by $\diag(x)$ the $2n$-vector containing its diagonalelements. This shows
    \begin{align*}
      \eta_{\{1,2\}\{3\}}(\omega m)	\geq	\br{\frac{2\omega n(n-1)}{L^3\cdot 48\mu n}}^{2/2} \geq \frac{\omega}{48 \mu L^3}\cdot n,
    \end{align*}
    where we used that $n-1\geq \frac{1}{2}n$ because $n\geq 2$.
    The other cases follow in a similar manner, the sums can be estimated directly or using the Cauchy Schwarz inequality by euclidian or $1$-norms of tensors with unit norm or by norms of their columns, rows or diagonal elements. We
    only state the results here:

    \begin{align*}
      \eta_{\{1\}\{2\}}(\omega m)	&\geq	\frac{\omega}{4L^2(\sigma^2 + 2 \mu^2)} \cdot n
 \\
      \eta_{\{1,2\}}(\omega m)	&\geq \frac{\omega^2}{32 L^4(\sigma^2 +  \mu^2 )^2}\cdot n\\
      \eta_{\{1\}\{2\}\{3\}}(\omega m)	&\geq\frac{\omega^{2/3}}{192^{2/3}\mu^{2/3} L^2}\cdot n\\
      \eta_{\{1,2,3\}} (\omega m) &\geq \frac{\omega^{2}}{48^2 \mu^2 L^6}\cdot n^2\\
      \eta_{\{1\}\{2\}\{3\}\{4\}}( \omega m)	&\geq		\frac{\omega^{1/2}}{24^{1/2}L^2}\cdot n\\
      \eta_{\{1,2\}\{3\}\{4\}}( \omega m)	&\geq	\frac{\omega^{2/3}}{64^{2/3} L^{8/3}}\cdot n\\
      \eta_{\{1,2\}\{3,4\}}( \omega m)	&\geq	\frac{\omega}{128L^{4}}\cdot n\\
      \eta_{\{1,2,3\}\{4\}}( \omega m)	&\geq\frac{\omega}{48 L^{4}}\cdot n\\
      \eta_{\{1,2,3,4\}}( \omega m)	&\geq \frac{\omega^2}{48^2L^{8}}\cdot n^2.
    \end{align*}

   \end{proof}

  \section{The $\psi_r$--norm via Moments}
  \label{app:psir:moments}
  It is well-known, see \cite{Vershynin:datasciencebook}, that
  \begin{equation}\label{eq:psi_r_equivalent}
    \norm{X}_{\psi_r} = \sup_{p\geq 1} p^{-1/r} (\expec{\abs{X}^p})^{1/p}
  \end{equation}
  is equivalent to \eqref{eq:def:psip}.  Now, let $a\in\C^n$ be a
  random vector with subgaussian entries and 
  $\norm{\re(a_i)}_{\psi_2},\norm{\im(a_i)}_{\psi_2}\leq \psi_2$ for a
  constant $\psi_2$. In this section we show how to estimate the
  $\psi_r$-norm for $r\geq 1$ of the matrix $aa^*$ by $\psi_2$. The $\psi_r$-norm of
  a random matrix $A\in\C^{n\times n}$ is defined as
  \begin{equation*}
    \norm{A}_{\psi_r}\coloneqq\sup_{\norm{Z}_F\leq 1}\norm{\scp{A}{Z}}_{\psi_r}.
  \end{equation*}
  For the matrix $aa^*-\expec{aa^*}$ this can be written as
  \begin{align*}
      \norm{aa^*-\expec{aa^*}}_{\psi_r} = \sup_{\norm{Z}_F\leq 1}\norm{\scp{a}{Za} - \expec{\scp{a}{ZA}}}_{\psi_r}.
  \end{align*}
  Set $Y_Z\coloneqq\scp{a}{Za} - \expec{\scp{a}{Za}}$ for some arbitrary $Z\in\C^{n\times n}$ with $0<\norm{Z}_F \leq 1$. Using \eqref{eq:psi_r_equivalent} we can compute its $\psi_r$-norm as 
  \begin{equation}\label{eq:psi_r_alternative}
      \norm{Y_Z}_{\psi_r} = c_r \cdot \sup_{p\geq 1} \frac{\expec{\abs{Y}^p}^{1/p}}{p^{1/r}},
  \end{equation}
  with some constant $c_r>0$.
  The expectation can be expressed as
  \begin{equation}\label{eq:LPintegral}
    \expec{\abs{Y_Z}^p} =   p\int_0^\infty  t^{p-1} \prob{ \abs{Y_Z} \geq t } \dif t.
  \end{equation}
  The Hanson-Wright inequality \eqref{eq:HWcomplex} yields
  \begin{align*}
      &\prob{\abs{Y_Z} \geq t} \\
      \leq& 4 \exp (-c \min \{ \frac{t^2}{4 \psi_2^4\norm{Z}^2T_{\fro} },\frac{t}{\sqrt{2}\psi_2^2 \norm{Z}_o}\})\\
      \leq& 4 \exp (-c \min \{ \frac{t^2}{4 \psi_2^4 },\frac{t}{\sqrt{2}\psi_2^2 }\})\\
      =& 4\max \bigset{e^{-t^2/a^2}, e^{-t/b}},
  \end{align*}
  where we used $\norm{Z_Y}_o\leq \norm{Z_Y}_F\leq 1$ and abbreviated
  $a\coloneqq\frac{2 \psi_2^2}{\sqrt{c}}$ and  $b\coloneqq \frac{\sqrt{2}\psi_2^2}{c}$. Plugging into  \eqref{eq:LPintegral} and substituting $s\coloneqq t/a$, respectively $s\coloneqq t/b$, we obtain
  \begin{align*}
     \expec{\abs{Y_Z}^p}  \leq&  4p \int_0^\infty s^{p-1} \max\smallset{a^p e^{-s^2}, b^p e^{-s} } \dif t\\
     \leq   &     4p\bigbr{\frac{1}{2}a^p \Gamma(\frac{p}{2}) + b^p \Gamma(p)},
  \end{align*}
  where we estimated the maximum by the sum of both terms and expressed the integrals in terms of the Gamma function. Using the identity $\Gamma(x)x = \Gamma(1+x)$ , for $x>0$, and the asymptotic estimation $\Gamma(x+1) \lesssim x^x$ derived from Stirling's formula, we obtain
  \begin{align*}
      \expec{\abs{Y_Z}^p}   &\leq 4 \bigbr{ a^p \Gamma(\frac{p}{2}+1) + b^p \Gamma(p+1) }\\
      &\leq c' \bigbr{a^p (\frac{p}{2})^{p/2} + b^p p^p}\\
      &\leq 2^{p/2} c' \psi_2^{2p} p^p \bigbr{ c^{-p/2} + c^{-p}},
  \end{align*}
  for some constant $c'>0$. Plugging this into \eqref{eq:psi_r_alternative} yields 
  \begin{align}
      \norm{Y_Z}_{\psi_1} &\leq \sqrt{2} c_r \psi_2^2 \cdot \sup_{p\geq 1} \bigbr{ c^{-p/2} + c^{-p}}^{1/p}\nonumber\\
      &= c'' \psi_2^2,\label{eq:psi_1_2}
  \end{align}
  where $c''$ is some constant that does not depend on the dimensions.
    Since for $r,p\geq 1$ it holds that $p^{-1/r}\leq p^{-1}$, we have $c_r^{-1} \normpsi{Y_Z}{r} \leq c_1^{-1} \normpsi{Y_Z}{1}$ whenever $r\geq 1$. Plugging into \eqref{eq:psi_1_2} and taking the supremum over all $Z\in\C^{n\times n}$ with $\norm{Z}_F \leq 1$ shows that $\normpsi{Y}{r} \leq c'_r \psi_2^2$, for some constant $c'_r$.

\end{appendices}

\bibliography{CS_nonneg}{}
\bibliographystyle{plainurl}

\end{document}